\documentclass[11pt]{article}
\pdfoutput=1 %Force pdflatex without ps compiler in the arXiv compiler 
\usepackage[margin=1.1in]{geometry}
\usepackage[hyphens]{url}  
\usepackage{graphicx} 
\usepackage[numbers]{natbib}  
\usepackage{caption}
\usepackage{subcaption} 

\title{A Multivariate Complexity Analysis\\of the Material Consumption Scheduling Problem\thanks{An extended abstract of this work appears in the \emph{Proceedings of the 35th~AAAI Conference on Artificial Intelligence, AAAI'21}. This full version contains additional results for multiple resources as well as full proof details.}}
\author {
        Matthias Bentert,\textsuperscript{\rm 1}
        Robert Bredereck,\textsuperscript{\rm 1, \rm 2}
		P\'{e}ter Gy\"{o}rgyi,\textsuperscript{\rm 3}\\
        Andrzej Kaczmarczyk,\textsuperscript{\rm 1}
	and Rolf Niedermeier\textsuperscript{\rm 1}\\[10pt]
 \small\textsuperscript{\rm 1} Technische Universit\"{a}t Berlin,
 Faculty~IV,\\\small Algorithmics and Computational Complexity,
 Berlin, Germany\\\small{\{matthias.bentert, a.kaczmarczyk, rolf.niedermeier\}@tu-berlin.de} \\[6pt]
 \small\textsuperscript{\rm 2} Humboldt-Universit\"{a}t zu Berlin, Institut
 f\"{u}r Informatik, Algorithm Engineering, Berlin, Germany \\\small{robert.bredereck@hu-berlin.de} \\[6pt]
 \small\textsuperscript{\rm 3} Institute for Computer Science and Control, E\"{o}tv\"{o}s Lor\'{a}nd Research Network, Budapest, Hungary\\\small{gyorgyi.peter@sztaki.hu}
}
\usepackage{amsmath, amsthm, amssymb,enumerate}
\usepackage[pagebackref]{hyperref}
\usepackage{cleveref}
\usepackage{tikz}
\usepackage{algorithm, algorithmicx, algpseudocode}
\usetikzlibrary{shapes,arrows,fit,calc,positioning,patterns}
\usepackage{nicefrac}
\usepackage{etoolbox}
\usepackage[textsize=scriptsize]{todonotes} %,disable
\usepackage{comment}
\usepackage{tabularx}
\usepackage{booktabs}

\newtheorem{theorem}{Theorem}
\newtheorem{lemma}{Lemma}

\newtheorem{asp}{Assumption}

\newtheorem{proposition}{Proposition}

\theoremstyle{definition}
\newtheorem{definition}{Definition}

\newcommand{\orderedlistingof}[2]{\ensuremath{#1_1, #1_2, \ldots, #1_#2}}

\newcommand{\unarybinpacking}{\textsc{Unary Bin Packing}}

\newcommand{\np}{\ensuremath{\mathrm{NP}}}

\newcommand{\fpt}{\ensuremath{\mathrm{FPT}}}                                     
\newcommand{\xp}{\ensuremath{\mathrm{XP}}}                                       
\newcommand{\wone}{\ensuremath{\mathrm{W[1]}}}                                   
\newcommand{\woneh}{\wone-h}
\newcommand{\wonehard}{\wone-hard}
\newcommand{\wonehardness}{\wonehard{}ness}

\newcommand{\p}{\ensuremath{\mathrm{P}}}

\newcommand{\nphard}{\np-hard}

\crefname{table}{Table}{Tables}
\crefname{figure}{Figure}{Figures}
\crefname{algorithm}{Algorithm}{Algorithms}
\crefname{theorem}{Theorem}{Theorems}
\crefname{definition}{Definition}{Definitions}
\crefname{corollary}{Corollary}{Corollaries}
\crefname{proposition}{Proposition}{Propositions}
\crefname{obs}{Observation}{Observations}
\crefname{lemma}{Lemma}{Lemmas}
\crefname{example}{Example}{Examples}
\crefname{reduction}{Reduction}{Reductions}
\crefname{algocf}{Algorithm}{Algorithms}

\newcommand{\J}{\mathcal{J}}
\newcommand{\R}{\mathcal{R}}
\newcommand{\cmax}{C_{\max}}
\newcommand{\pmax}{p_{\max}}
\newcommand{\amax}{a_{\max}}
\newcommand{\bmax}{b_{\max}}
\newcommand{\umax}{u_{\max}}
\newcommand{\noGapProblem}{\ensuremath{1, \NI |nr=r| - }}
\newcommand{\noGapProblemOne}{\ensuremath{1, \NI |nr=1| - }}
\DeclareMathOperator{\nr}{nr}

\DeclareMathOperator{\NI}{NI}
\DeclareMathOperator{\const}{const}
\DeclareMathOperator{\unary}{unary}

\newcommand{\classP}{\ensuremath{\mathsf{P}}}
\newcommand{\NP}{\ensuremath{\mathsf{NP}}}

\newcommand{\paraNP}{para-\NP}
\newcommand{\pnp}{p-\NP}
\newcommand{\FPT}{\ensuremath{\mathsf{FPT}}}

\newcommand{\XP}{\ensuremath{\mathsf{XP}}}

\newcommand{\probDef}[4]{
\begin{center}   
	\begin{minipage}{.95\textwidth}
		\vspace{2pt} 

		\noindent
		\normalsize\textsc{#1}
		
		\vspace{1pt}
		\setlength{\tabcolsep}{3pt}
		\renewcommand{\arraystretch}{1.0}
		\begin{tabularx}{\textwidth}{@{}lX@{}}
			\normalsize\textbf{Input:} 	& \normalsize#2 \\
			\normalsize\textbf{#4:} 	& \normalsize#3
		\end{tabularx}
	\end{minipage}
\end{center}
}

\newcommand{\optProb}[3]{\probDef{#1}{#2}{#3}{Task}}

\newcommand{\probMCSLong}{\textsc{Material Consumption Scheduling Problem}}

\begin{document}
\maketitle

\begin{abstract}
 The NP-hard \textsc{Material Consumption Scheduling Problem} and 
 related problems have been thoroughly studied since the 1980's. Roughly
 speaking, the problem deals with minimizing the makespan when scheduling jobs
 that consume non-renewable resources. We focus on the single-machine case
 without preemption: from time to time, the resources of the machine are
 (partially) replenished, thus allowing for meeting a necessary precondition for
 processing further jobs, each of which having individual resource demands. We initiate a
 systematic exploration of the parameterized computational complexity landscape of the
 problem, providing parameterized tractability as well as intractability
 results. Doing so, we mainly investigate how parameters related to the resource
 supplies influence the problem's computational complexity. This leads to a deepened
 understanding of this fundamental scheduling
 problem.
\end{abstract}


\section{Introduction}
Consider the following motivating example. 
Every day, an agent works for a number of clients, all of equal importance. 
The clients, one-to-one corresponding to jobs,
each time request a service having individual processing time and 
individual consumption of a non-renewable resource; examples 
for such resources include raw material, energy, and money. 
The goal is to finish all jobs as early as possible, 
known as minimizing the makespan in the scheduling literature.
Unfortunately, the agent only has a limited 
initial supply of the resource which is to be 
renewed (with potentially different amounts)
at known points of time during the day. 
Since the
job characteristics (resource consumption, job length) and 
the resource delivery characteristics (delivery amount, point of time)
are known in advance, 
the objective thus is to find a feasible 
job schedule minimizing the makespan. Notably, jobs 
cannot be preempted and only one at a time can be executed.
Figure~\ref{fig:ExampleSchedule} provides a concrete numerical 
example with 
six jobs having varying job lengths and resource 
requirements.
\begin{figure*}
\begin{subfigure}{.3\textwidth}
\center
\begin{small}
\begin{tabular}{c|c|c|c|c|c|c}
$p_j$ & 1 & 1 & 1 & 2 & 2 & 3 \\ 
\midrule
{$a_{j}$} & 3 & 1 & 2 & 3 & 2 & 6 \\ 
\end{tabular}

\vspace{.4cm}
\begin{tabular}{c|c|c|c|c}
$u_\ell$ & 0 & 3 & 5 & 9\\ 
\midrule
$\tilde{b}_\ell$ & 3 & 6 & 2 & 6\\ 
\end{tabular}
\end{small}
\end{subfigure}
\begin{subfigure}{.7\textwidth}
\begin{tikzpicture}[thick,scale=.8, every node/.style={transform shape}]
\def\yy{-2}
\coordinate(u1) at (0,\yy);
\coordinate(u2) at (3,\yy);
\coordinate(u3) at (5,\yy); 
\coordinate(u4) at (9,\yy);
\coordinate(cm) at (12,\yy);

\coordinate(uv1) at (0,\yy-0.2);
\coordinate(uv2) at (3,\yy-0.2);
\coordinate(uv3) at (5,\yy-0.2); 
\coordinate(uv4) at (9,\yy-0.2);
\coordinate(cmv) at (12,\yy-0.2);

\coordinate(dd) at (-0.3,\yy);
\draw [-latex](0,\yy)-- (12.5,\yy)coordinate(ff) node[above]{$t$}; 

\draw[] (u1) -- (uv1) node[below] {$u_{1}$};

\foreach \xx in{2,3,4}{
\draw[] (u\xx) -- (uv\xx) node[below] {$u_{\xx}$};
}
\draw[] (cm) -- (cmv) node[below] {$\cmax=12$};

\def\hhh{0.29}
\def\hh{\hhh+\hhh}
\node[above right=\hhh cm and -0.01 cm of u1,right,draw, minimum width=1cm,minimum height=\hh cm,fill=white](j1) {$J_3$};
\node[right=-0.1cm of j1,right, draw,minimum width=1cm,minimum height=\hh cm,fill=white](j2){$J_2$};

\node[above right=\hhh cm and -0.01cm of u2,right, draw,minimum width=1.1cm,minimum height=\hh cm,fill=white](j3){$J_1$};
\node[right=-0.1cm of j3,right, draw,minimum width=2cm,minimum height=\hh cm,fill=white](j4){$J_5$};
\node[right=-0.1cm of j4,right, draw,minimum width=2cm,minimum height=\hh cm,fill=white](j5){$J_4$};

\node[above right=\hhh cm and -0.01cm of u4,right, draw,minimum width=3cm,minimum height=\hh cm,fill=white](j6){$J_6$};
\end{tikzpicture}
\end{subfigure}
\caption{An example (left) with one resource type and a solution (right) with makespan~$12$.
	The processing times  and the resource requirements are in the first table, while the supply dates and the supplied quantities are in the second.
	Note that $J_3$ and $J_2$ consume all of the resources supplied at $u_1=0$, thus we have to wait for the next supply to schedule further jobs.}
\label{fig:ExampleSchedule}
\end{figure*}
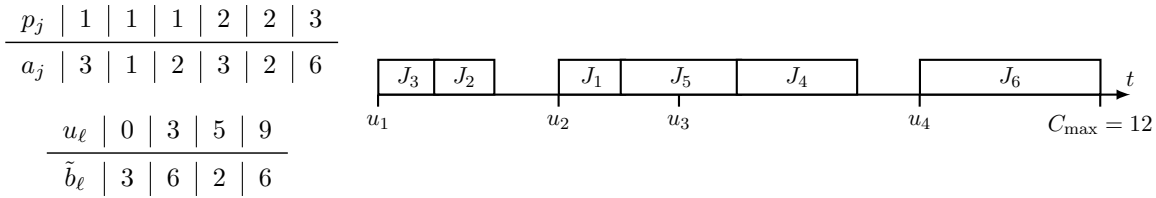%

The described problem setting
is known as minimizing the makespan on
a single machine with non-renewable resources.
More specifically, we study the single-machine variant of the NP-hard 
\probMCSLong{}.
Formally, we study the following problem.
\optProb{Material Consumption Scheduling Problem}
{
	A set~$\R$ of resources, a set~$\J=\{J_1,\ldots,J_n\}$ of jobs each job~$J_j$ with a processing time~$p_j\in \mathbb{Z}_+$ and a resource requirement $a_{ij}\in \mathbb{Z}_+$ from resource~$i\in \R$, and a set~$\{u_1,u_2,\ldots,u_q\}$ of points of time with~$0=u_1<u_2<\cdots<u_q$, where~$\tilde{b}_{i,\ell}$~quantities of resource~$i\in \R$ are supplied.
}
{
Find a \emph{schedule}~$\sigma$ with minimum makespan for a single machine without preemption which is \emph{feasible}, that is, (i) the jobs do not overlap in time, and (ii) at any point of time~$t$ the total supply from each resource is at least the total request of the jobs starting until~$t$.
}

The objective is to minimize the makespan, that is, the maximum time that a job is completed.
Formally, the makespan is defined by~$\cmax:=\max_{j\in \J}C_j$, where $C_j$ is the completion time of job $J_j$.
Notably, in our example in Figure~\ref{fig:ExampleSchedule} we considered the special but perhaps most prominent case of just one type of resource.
In this case we simply drop the indices corresponding to the single resource.
In the remainder of the paper, we make the following simplifying assumptions
guaranteeing sanity of the instances and filtering out trivial cases.
\begin{asp}\label{asp:1}
 Without loss of generality, we assume that
 \begin{enumerate}
  \item there are enough resources supplied to process all jobs:
   $\sum_{\ell=1}^q\tilde{b}_\ell\geq \sum_{j\in\J}a_j$;
  \item each job has at least one non-zero resource requirement: 
   $\forall_{j \in \J} \sum_{i \in \R} a_{i,j} > 0$; and
  \item at least one resource unit is supplied at time~$0$:
   $\sum_{i \in \R} \tilde{b}_{i, 0} > 0$.
 \end{enumerate}
\end{asp}

Note that each of these assumptions can be verified in linear time.
It is valid to make these assumptions because of the following.
If the first assumption does not hold, then there is no feasible schedule.
If the second assumption does not hold, then we can schedule all jobs without resource requirements in the beginning.
Thus, we can remove these jobs from the instance and adjust the supply times of new resources accordingly.
If the third assumption does not hold (but the second does), then we cannot schedule any job before the first resource supply.
Thus, we can adjust the supply time of each resource requirement such that there is a resource supply at time~$0$ and get an equivalent instance. 

It is known that the \probMCSLong{} is NP-hard even in the case of just one machine,
only two supply dates ($q=2$),
and if the processing time of each job is the same as its resource requirement,
that is, $p_j=a_j$ for each~$j\in \J$ \cite{Carlier84}.
While many variants of the \probMCSLong{} have been studied in the literature in terms of heuristics,
polynomial-time approximation algorithms, or the detection of 
polynomial-time solvable special cases, we are not aware 
of any previous systematic studies 
concerning a multivariate complexity analysis.
In other words, we study, seemingly for the first time, several 
natural problem-specific parameters
and investigate how they influence the computational complexity of the problem.
Doing so, we prove both parameterized hardness as well as 
fixed-parameter tractability results for this NP-hard problem.

\paragraph{Related Work.}
Over the years, performing multivariate, parameterized 
complexity studies for
fundamental scheduling problems became more and more 
popular~\cite{BBN19,BMNW15,BNS17,BF95,BW20,GHM20,FM03,HHMMNS21,HKSTW19,HPST19,HMPSY20,HST19,KK18,MB18,MW15}.
We contribute to this field by a seemingly first-time exploration of 
the \probMCSLong, focusing on one machine 
and the minimization of the makespan.

The problem was introduced in the 1980's~\cite{Carlier84,Slowinski84}. Indeed, even a bit earlier a problem where the jobs required non-renewable resources, but without any machine environment, was studied~\cite{CarlRin82}. 
There are several real-world applications, for instance, in the
continuous casting stage of steel production~\cite{herr16}, in managing
deliveries by large-scale distributors~\cite{Belkaid12}, or in shoe production~\cite{Carrera10}.

\citet{Carlier84} proved several complexity results for different variants in the single-machine case,
while \citet{Slowinski84} studied the parallel machine variant of the problem with preemptive jobs.
Previous theoretical results mainly concentrate on the computational complexity and polynomial-time approximability of different variants; in this literature review we mainly focus on the most important results for the single-machine case and minimizing makespan as the objective.
We remark that there are several recent results for variants with other objective functions~\cite{Berczi20,Gyorgyi19,gyorgyi20}, with a more complex machine environment~\cite{Gyorgyi17}, and with slightly different resource constraints~\citep{davari20}.

\citet{Toker91} proved that the variant where the jobs require one non-renewable resource reduces to the \textsc{2-Machine Flow Shop Problem} provided that the single non-renewable resource has a 
unit supply in every time period.
Later, \citet{Xie97} generalized this result to multiple resources.
\citet{Grigoriev05} showed that the variant with unit processing times and two resources is NP-hard, and they also provided several polynomial-time 2-approximation algorithms for the general problem. 
There is also a polynomial-time approximation scheme (PTAS) for the variant with one resource and a constant number of supply dates and a fully polynomial-time approximation scheme (FPTAS) for the case with~$q=2$ supply dates and one non-renewable resource~\cite{Gyorgyi14}. 
\citet{Gyorgyi15}
presented approximation-preserving reductions between problem variants with~$q=2$ and variants of the \textsc{Multidimensional Knapsack Problem}.
These reductions have several consequences, for example,  it was shown that the problem is NP-hard if there are two resources, two supply dates, and each job has a unit processing time, or that there is no FPTAS for the problem with two non-renewable resources and $q=2$ supply dates, unless P $=$~NP. 
Finally, there are three further results~\cite{Gyorgyi15b}: 
(i)~a PTAS for the variant where the number of resources and the number of
supply dates are constants; (ii)~a PTAS for the variant with only one resource and an arbitrary number of supply dates if  the resource requirements are proportional to job processing times; and (iii)~an APX-hardness when the number of resources is part of the input. 

\paragraph{Preliminaries and Notation.}
We use the standard three-field $\alpha|\beta|\gamma$-notation~\cite{Graham79},
where $\alpha$~denotes the machine environment, $\beta$~the further constraints
like additional resources, and $\gamma$~the objective function. We always
consider a single machine, that is, there is a~$1$ in the $\alpha$ field. The non-renewable
resources are described by~$\nr$ in the $\beta$ field and $\nr = r$ means that
there are $r$ different resource types. In our work, 
the only considered objective is the
makespan~$\cmax$. The~\probMCSLong{} variant with a single
machine, single resource type, and with the makespan as the objective is then expressed as~$1| \nr = 1|\cmax$. Sometimes, we also consider the so-called
\emph{non-idling scheduling} (introduced by~\citet{C08:non-idle}), indicated by $\textrm{NI}$ in the $\alpha$ field, in which a
machine can only process all
jobs continuously, without intermediate idling. As we make the simplifying
assumption that the machine has to start processing jobs at time~$0$, we drop
the optimization goal~$\cmax$ whenever considering non-idling scheduling.
When there is just one resource ($\nr=1$), then we write~$a_j$ instead
of~$a_{1,j}$ and~$\tilde{b}_{j}$ instead of~$\tilde{b}_{1,j}$, etc.
We also write~$p_j=1$ or $p_j=ca_j$ whenever, respectively, jobs have solely unit
processing times or the resource requirements are proportional to the job
processing times. Finally, we use~``$\unary$'' to indicate that all numbers in
an instance are encoded in unary. Thus, for example, $1, \NI| p_j=1, \unary | -$
denotes a single non-idling machine, unit-processing-time jobs and the unary
encoding of all numbers. We summarize the notation of the parameters that we
consider in the following table.
\begin{center}
\begin{tabular}{l|l}
\toprule
$n$ & number of jobs \\
$q$ & number of supply dates\\
$j$ & job index\\
$\ell$ & index of a supply\\
$p_j$ & processing time of job $j$\\
$a_{i,j}$ &  resource requirement of job $j$ from resource $i$\\
$u_\ell$ &  the  $\ell^{\text{th}}$ supply date\\
$\tilde{b}_{i,\ell}$ & quantity supplied from resource $i$ at  $u_\ell$\\
$b_{i,\ell}$ & total resource supply from resource $i$\\ &over the first $\ell$
supplies, that is,  $\sum_{k=1}^\ell \tilde{b}_{i, k}$\\
\toprule
\end{tabular}
\end{center}
To simplify matters, we introduce the shorthands~$\amax$,~$\bmax$, and~$\pmax$
for~$\max_{j\in \J, i\in \R}a_{ij}$, $\max_{\ell\in \{1,\ldots,q\}, i\in \R}\tilde{b}_{i, \ell}$, and~$\max_{j\in \J}p_j$, respectively.

\paragraph{Primer on Multivariate Complexity.}
To analyze the parameterized complexity~\cite{CFKLMPPS15,DF13,FG06,Nie06}
of the \probMCSLong{}, we declare some part of the input the \emph{parameter}
(e.g., the number of supply dates).
A parameterized problem is \emph{fixed-parameter tractable}
if it is in the class \FPT{} of problems solvable in~$f(\rho) \cdot
|I|^{O(1)}$~time, where $|I|$~is the size of a given instance encoding,
$\rho$~is the value of the parameter, and $f$~is an arbitrary computable
(usually super-polynomial) function.
Parameterized hardness (and completeness) is defined through parameterized reductions
similar to classical polynomial-time many-one reductions.
For our work, it suffices to additionally ensure that the value of the parameter in the
problem we reduce to depends only on the value of the parameter of the
problem we reduce from.
To obtain parameterized intractability, we use parameterized reductions
from problems of the class~$\wone$ which is widely believed 
to be a proper superclass of~\FPT. For instance, the famous graph problem \textsc{Clique} is $\wone$-complete with respect to the parameter size of the clique~\cite{DF13}.

The class~\XP{} contains
all problems that can be solved in~$|I|^{f(\rho)}$~time for a function~$f$
solely depending on the parameter~$\rho$.
While \XP{} ensures polynomial-time solvability when~$\rho$ is a constant,
\FPT{} additionally ensures that the degree of the polynomial is independent
of~$\rho$.
Unless~$\classP=\NP$, membership in \XP{} can be excluded by showing that the problem
is \NP-hard for a constant parameter value---for short, we say that the problem is \paraNP-hard.

\paragraph{Our Contributions.}

Most of our results are summarized in Table~\ref{tab:results}.
We focus on the parameterized computational complexity of the \probMCSLong{} with
respect to several parameters describing resource supplies.
We show that the case of a single resource and jobs with unit processing time is polynomial-time solvable.
However, if each job has a processing time proportional to its resource
requirement, then the \probMCSLong{} becomes \nphard{} even for a single resource
and when each supply provides one unit of the resource.
Complementing an algorithm solving the \probMCSLong{} in polynomial time for a constant number~$q$ of supply dates, we show by proving \wonehardness{},
that the parameterization by~$q$ presumably does not yield fixed-parameter
tractability.
We circumvent the~\wonehardness{} by combining the parameter~$q$ with the maximum resource requirement~$\amax$ of a job,
thereby obtaining fixed-parameter tractability for the combined parameter~$q+\amax$.
Moreover, we show fixed-parameter tractability for the parameter~$u_{\max}$ which denotes the last resource supply time.
Finally, we provide an outlook on cases with multiple resources and show
that fixed-parameter tractability for~$q+\amax$ extends when we additionally add the number of resources~$r$ to the combined parameter, that is, we show fixed-parameter tractability for~$q+\amax+r$.
For the \probMCSLong{} with an unbounded number of resources, we show intractability even for the case
where all other previously discussed parameters are combined.

\begin{table*}[t]
 \centering
 \caption{Our results for a single resource type (top) and multiple
  resource types (bottom). The results correspond to
  \Cref{thm:unit-P}~($\ddag$),
 \Cref{thm:q-w1h}~($\diamond$),
 \citet{Gyorgyi14}~($\clubsuit$),
 \Cref{thm:b-pNP-hard}~($\blacksquare$),
 \Cref{thm:umax-fpt}~($\Diamond$),
 \Cref{thm-amax-XP}~($\blacktriangle$),
 \Cref{thm:q-bmax-fpt}~($\dagger$),
 \Cref{prop:nr2qW1}~($\heartsuit$),
 \Cref{prop:q-bmax-fpt-const-resources}~($\spadesuit$),
 \Cref{thm-superhard}~($\blacktriangledown$), and
 \Cref{prop-umax-XP}~($\blacklozenge$).
	\p{} stands for polynomial-time solvable, \woneh{} and \pnp{} stand for W[1]-hardness and \paraNP-hardness, respectively. 
 }
 \begin{tabular}{lccccc}
   & $q$ & $b_{\max}$ & $u_{\max}$ & $a_{\max}$ & $a_{\max} + q$
  \\ \toprule

  $1|\nr=1,\ p_j = 1 | \cmax$ &
  \multicolumn{5}{c}{\p{}$^\ddag$} \\

  $1 |\nr=1,\ p_j = c a_j | \cmax$ &
  \woneh{}$^\diamond$, \xp{}$^\clubsuit$ & \pnp{}$^\blacksquare$ &\fpt{}$^\Diamond$ & \xp{}$^\blacktriangle$ & 
  \fpt$^\dagger$ \\

  $1 |\nr=1,\ \unary | \cmax$ &
  \woneh{}$^\diamond$, \xp{}$^\clubsuit$ & \pnp{}$^\blacksquare$ & \fpt{}$^\Diamond$ & \xp{}$^\blacktriangle$ &
  \fpt$^\dagger$ \\ \midrule

  $1|\nr=2,\ p_j = 1,\ \unary | \cmax$ &
  \woneh{}$^\heartsuit$, \xp{}$^\clubsuit$ & \pnp{}$^\blacksquare$ & \xp{}$^\clubsuit$ & \xp{}$^\blacktriangle$ &
  \fpt$^\spadesuit$\\

  $1 |\nr= \const,\ \unary | \cmax$ &
  \woneh{}$^\heartsuit$, \xp{}$^\clubsuit$ & \pnp{}$^\blacksquare$ & \xp{}$^\clubsuit$ & \xp{}$^\blacktriangle$ &
  \fpt$^\spadesuit$\\

  $1 |\nr,\ p_j=1| \cmax$ &
  \pnp{}$^\blacktriangledown$ & \pnp{}$^\blacktriangledown$ & \woneh{}$^\blacktriangledown$, \xp{}$^\blacklozenge$ & \pnp{}$^\blacktriangledown$ &
  \woneh{}$^\blacktriangledown$\\ \bottomrule

 \end{tabular}

 \label{tab:results}
\end{table*}

\section{Computational Hardness Results}
We start our investigation on the \probMCSLong{}
with outlining the limits of efficient computability. Setting up clear borders
of tractability, we identify potential scenarios suitable for spotting efficiently solvable special cases.
This approach is especially justified because the \probMCSLong{} is
already~\nphard{} for the quite constrained scenario of unit processing times and
two resources~\citep{Grigoriev05}.

Both hardness results in this section use reductions from \textsc{Unary Bin Packing}.
Given a number~$k$ of bins, a bin size~$B$, and a set~$O=\{o_1,o_2,\ldots o_n\}$
of~$n$~objects
of~sizes~$s_1,s_2,\ldots s_n$ (encoded in unary),
\textsc{Unary Bin Packing} asks to distribute the objects to the bins such that
no bin exceeds its capacity.
\textsc{Unary Bin Packing} is \NP-hard and \wone-hard parameterized by the number~$k$ of bins
even if~$\sum_{i=1}^{n} s_i = k B$ \cite{JKMS13}.

We first focus on the case of a single resource, for which we find a strong
intractability result. In the following theorem, we show that even if each
supply comes with a single unit of a resource, then the problem is
already~\nphard{}.

\begin{theorem} \label{thm:b-pNP-hard}
 $1|\nr=1, p_j = c a_j|\cmax$ is para-\nphard{} with respect to the maximum
 number~$b_{\max}$ of resources supplied at once even if all numbers are encoded
 in unary.
\end{theorem}
\begin{proof}
 Given an instance~$I$ of \textsc{Unary Bin Packing} with~$\sum_{i=1}^{n} s_i = k B$, 
 we construct an instance~$I'$ of~$1|\nr=1|\cmax$ with~$\bmax = 1$ as described below. 

 We define~$n$~jobs~$J_1=(p_1,a_{1}),J_2=(p_2,a_{2}),\ldots, J_n=(p_n,a_{n})$
 such that~$p_i = 2B s_{i}$ and $a_i = 2s_i$. We also introduce a special job~$J^* =
 (p^*, a^*)$, with $p^*=2B$ and~$a^*=1$. Then, we set $2kB$~supply dates as
 follows. For each~$i \in \{0, 1, \ldots, k-1\}$ and~$x \in \{0, 1, \ldots, 2B-1\}$, we
 create a supply date~$q_i^x = (u_i^x, \tilde{b}_i^x) := ((2B + i2B^2) - x,1)$.
 We add a special supply date~$q^* := (0, 1)$. Next, we show that~$I$ is a
 yes-instance if and only if there is a gapless schedule for~$I'$, that
 is,~$C_{\max} = 2(B^2+B)$.
 An example of this construction is depicted 
in~\cref{fig:unaryBinPacking}. 

 \begin{figure*}
\centering
\begin{tikzpicture}
	\def\unit{1.05cm}
	\begin{scope}[x=\unit]
	 \node[rectangle,pattern=north east lines, minimum height=10pt, minimum
	 width=1.5*\unit, anchor = west] at (-3,5.5) () {};
	 \node[rectangle,pattern=north east lines, minimum height=10pt, minimum
	 width=1.4*\unit, anchor=west] at (3.1,5.5) () {};

	 \node[rectangle, draw, minimum width=1.5*\unit,
	 label={\small\vphantom{$J_1$}$J^*$}] at (-2.25,6) {};
 	\node[rectangle, draw, minimum width=1.5*\unit, label={\small $J_1 = (8, 1)$}] at (-.75,6) {};
         \node[rectangle, draw, minimum width=4.5*\unit, label={\small $J_4 = (24, 8)$}] at (2.25,6) {};
         \node[rectangle, draw, minimum width=3*\unit, label={\small $J_2 = (16, 2)$}] at (6,6) {};
 	\node[rectangle, draw, minimum width=3*\unit, label={\small $J_3 = (16, 1)$}] at (9,6) {};
 
 	\node at (-3,5) {$0$};
 	\node at (-1.5,5) {$8$};
 	\node at (0,5) {$16$};
 	\node at (1.5,5) {$24$};
 	\node at (3,5) {$32$};
 	\node at (4.5,5) {$40$};
 	\node at (6,5) {$48$};
 	\node at (7.5,5) {$56$};
 	\node at (9,5) {$64$};
 	\node at (10.5,5) {$72$};
 
 	\node at (-3.5,5.5) (a) {};
 	\node at (11,5.5) (b) {};
 	\node at (-3,5.15) (c) {};
 	\node at (-3,5.85) (d) {};
 	\node at (-1.5,5.15) (e) {};
 	\node at (-1.5,5.85) (f) {};
 	\node at (0,5.25) (g) {};
 	\node at (0,5.75) (h) {};
 	\node at (1.5,5.25) (i) {};
 	\node at (1.5,5.75) (j) {};
 	\node at (3,5.25) (k) {};
 	\node at (3,5.75) (l) {};
 	\node at (4.5,5.15) (m) {};
 	\node at (4.5,5.85) (n) {};
 	\node at (6,5.25) (o) {};
 	\node at (6,5.75) (p) {};
 	\node at (7.5,5.25) (q) {};
 	\node at (7.5,5.75) (r) {};
 	\node at (9,5.25) (s) {};
 	\node at (9,5.75) (t) {};
 	\node at (10.5,5.25) (u) {};
 	\node at (10.5,5.75) (v) {};
 	\draw[->] (a) to (b);
 	\draw (c) to (d);
 	\draw (e) to (f);
 	\draw (g) to (h);
 	\draw (i) to (j);
 	\draw (k) to (l);
 	\draw (m) to (n);
 	\draw (o) to (p);
 	\draw (q) to (r);
 	\draw (s) to (t);
 	\draw (u) to (v);
       \end{scope} 
\end{tikzpicture}
\caption{An example of the construction in the proof of \cref{thm:b-pNP-hard} for an
instance of \textsc{Unary Bin Packing} consisting of~$k=2$~bins each of
size~$B=4$ and four objects~$o_1$ to~$o_4$ of sizes $s_1=1$, $s_2=s_3=2$,
and~$s_4=3$. In the
resulting instance of~$1|\nr=1, p_j = c a_j|\cmax$, there are five
jobs~($J^*$ and one job corresponding to each input object) and at each (whole)
point in time of
the hatched periods there is a supply of one resource. An optimal
schedule that first schedules~$J^*$ is depicted. Note that the time periods
between the (right-hand) ends of hatched periods correspond to a multiple of the bin size and
a schedule is gapless if and only if the objects corresponding to jobs scheduled
between the ends of two consecutive shaded areas exactly fill a bin.}
\label{fig:unaryBinPacking}
\end{figure*}
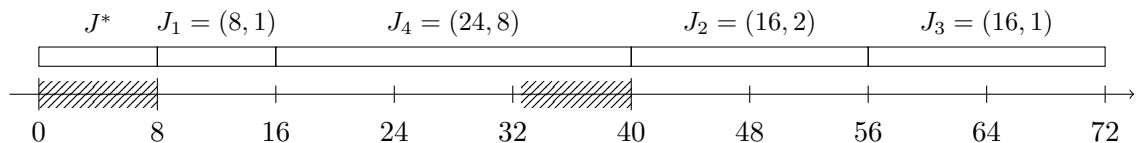
 
 We first show that each solution to~$I$ can be efficiently transformed to a
 schedule with~$C_{\max} = 2(B^2+B)$. A yes-instance for~$I$ is a partition of
 the objects into~$k$ bins such that each bin is (exactly) full. Formally, there
 are $k$~sets~$S_1,S_2,\ldots S_k$ such that~$\bigcup_i S_i = O$, $S_i \cap S_j
 = \emptyset$ for all~$i \neq j$, and~$\sum_{o_i \in S_j} s_i = B$ for all~$j$.
 We form a schedule for~$I'$ as follows. First, we schedule job~$J^*$ and then,
 continuously, all jobs corresponding to elements of set~$S_1$, $S_2$, and so
 on. The special supply~$q^*$ guarantees that the resource requirement of
 job~$J^*$ is met at time~$0$. The remaining jobs, corresponding to
 elements of the partitions, are scheduled earliest at time~$2B$,
 when~$J^*$ is processed. The jobs representing each partition, by definition,
 require in total $2B$~resources and take, in total, $2B^2$ time. Thus, it is
 enough to ensure that in each point~$2B + i2B^2$, for~$i \in \{0, 1, \ldots,
 k-1\}$, there are at least $2B$~resources available. This is true
 because for each~$i \in \{0, 1, \ldots, k-1\}$ the time point~$2B + iB^2$ is preceded with $2B-1$~supplies
 of one resource. Furthermore, none of the preceding jobs can use the
 freshly supplied resources as the schedule must be gapless and all
 processing times are multiples of~$2B$. As a result, the schedule is
 feasible.

 Now we show that a gapless schedule for~$I'$ implies that~$I$ is a
 yes-instance. Let $\sigma$ be a gapless schedule for~$I'$. Observe that
 all processing times are multiples of~$2B$ and therefore each job has to
 start at a time that is a multiple of~$2B$. For each~$i \in \{0, 1, \ldots,
 k-1\}$, we show that there is no job that is scheduled to start before~$2B + iB^2$ and
 to end after this time. We show this by induction on~$i$. Since at time~$0$ there
 is only one resource available, job~$J^*$ (with processing time~$2B$) must be
 scheduled first. Hence the statement holds for~$i=0$. Assuming that the
 statement holds for all~$i < i'$ for some~$i'$, we show that it also holds
 for~$i'$. Assume towards a contradiction that there is a job~$J$ that starts
 before~$t := 2B + i'B^2$ and ends after this time. Let~$S$ be the set of all jobs that
 were scheduled to start between~$t_0 := 2B + (i'-1)B^2$ and~$t$. Recall that for
 each job~$J_{j'} \in S$, we have that~$p_{j'} = a_{j'}B$. Hence, since~$J$ ends
 after~$t$, the number of resources used by~$S$ is larger than~$\nicefrac{(t -
 t_0)}{B} = B$. Since only $2B$~resources are available at time~$t$, job~$J$ cannot be
 scheduled before time~$t$ or there is a gap in the schedule, a
 contradiction. Hence, there is no job that starts before~$t$ and ends after it.
 Thus, the jobs can be partitioned into ``phases,'' that is, there are~$k+1$
 sets~$T_0, T_1, \ldots, T_k$ such that $T_0 = \{J^*\}$, $\bigcup_{h>0} T_h = \mathcal{J}
 \setminus \{J^*\}$, $T_h \cap T_j = \emptyset$ for all~$h \neq j$,
 and~$\sum_{J_j \in T_g} p_j = 2B^2$ for all~$g$. This corresponds to a bin
 packing where~$o_g$ belongs to bin~$h>0$ if and only if~$J_g \in T_h$.
\end{proof}

Note that \Cref{thm:b-pNP-hard} excludes pseudo-polynomial algorithms for the case under
consideration since the theorem statement is true also when all numbers are
encoded in unary.
\Cref{thm:b-pNP-hard} motivates to study further problem-specific parameters. 
Observe that in
the reduction presented in the proof of \Cref{thm:b-pNP-hard}, 
we used an unbounded number of supply dates.
\mbox{\citet{Gyorgyi14}} have shown a pseudo-polynomial algorithm for $1|\nr=1|\cmax$
for the case that the number~$q$ of supplies is a constant.
Thus, the question arises whether we can even obtain fixed-parameter tractability
for our problem by taking the number of supply dates as a parameter.
Devising a reduction from \textsc{Unary Bin Packing},
we answer this question negatively in the following theorem.

\begin{theorem} \label{thm:q-w1h}
 $1|\nr=1,\ p_j=a_j|\cmax$ parameterized by the number~$q$ of supply dates is
 \wonehard{} even if all numbers are encoded in unary.
\end{theorem}
\begin{proof}
 We reduce from \textsc{Unary Bin Packing} parameterized by the number~$k$ of
 bins, which is known to be W[1]-hard \cite{JKMS13}. Given an instance~$I$ of
 \textsc{Unary Bin Packing} with~$k$ bins, each of size~$B$ and~$n$
 objects~$o_1,o_2,\ldots o_n$ of size~$s_1,s_2,\ldots s_n$ such
 that~$\sum_{i=1}^{n} s_i = k B$, we construct an instance~$I'$
 of~$1|\nr=1|\cmax$ as follows. We denote the set of all objects by~$O$.

 First, for each object~$o_i \in O$, we define a job~$J_i=(p_i,a_i)$ such that $p_i =
 a_{i} = s_i$; we denote the set of all jobs by~$\mathcal{J}$.
 Next, we construct $k$~supply dates~$q_i=(u_i,\tilde{b}_{i})$, with~$u_i = (i-1)B$
 and~$\tilde{b}_{i} = B$ for each~$i \in [k]$.
 This way we obtain an instance~$I'$ of~$1|\nr=1|\cmax$.

 It remains to show that~$I$ is a yes-instance if and only if~$I'$ is an
 instance with~$\cmax = kB$. To this end, suppose first that~$I$ is a
 yes-instance. Then, there is a partition of the objects into $k$~bins such that
 each bin is (exactly) full. Formally, there are~$k$ sets~$S_1,S_2,\ldots S_k$
 such that~$\bigcup_i S_i = O$, $S_i \cap S_j = \emptyset$ for all~$i \neq j$,
 and~$\sum_{o_i \in S_j} s_i = B$ for all~$j$.
 Hence we schedule all jobs~$j_i$ with~$o_i \in S_1$ between time~$0 = u_1$ and~$B = u_2$.
 Following the same procedure, we schedule all jobs corresponding to objects
 in~$S_i$ between time~$u_i$ and~$u_{i+1}$ where~$u_{i+1} = i B$.
 Since~$\sum_{i=1}^{n} p_i = k B$, we conclude that~$\cmax = kB$.

 Now suppose that~$\cmax = kB$.
 Assume towards a contradiction that there is an optimal schedule in which some
 job~$J_i \in \mathcal{J}$ starts at some time~$t$ such that~$t < u_\ell < t + p_i$
 for some~$\ell \in [k]$.
 Let~$S$ be the set of all objects that are scheduled before~$J_i$.
 Since~$\cmax = kB$, it follows that at each point of time until~$t$,
 there is some job scheduled at this time.
 Thus, since~$p_h = a_{h}$ for all jobs~$J_h$, it follows that~$\sum_{J_h \in S}
 a_{h} = \sum_{J_h \in S} p_{h} = t$. As a result, $\sum_{J_h \in S \cup \{J_i\}}
 a_{h} = t + a_{i} = t + p_i > u_\ell = (\ell-1) B = \sum_{h=1}^{\ell-1}
 \tilde{b}_{h}$; a contradiction to~$j_i$ starting before~$u_\ell$.
 Hence,  there is no job that starts before some~$u_\ell$ and ends after it.
 Thus, the jobs can be partitioned into ``phases,'' that is, there are~$k$
 sets~$T_1,T_2,\ldots, T_k$ such that~$\bigcup_h T_h = \mathcal{J}$,
 $T_i \cap T_{i'} = \emptyset$ for all~$i \neq i'$,
 and~$\sum_{J_h \in T_g} p_h = B$ for all~$g$.
 This corresponds to a bin packing where~$o_g$ belongs to bin~$h$ if and only
 if~$J_g \in T_h$.
\end{proof}

The theorems presented in this section show that our problem is
(presumably) not fixed-parameter tractable either with respect to the number of
supply dates or with respect to the maximum number of resources per supply.
However, as we show in the following section, combining these two parameters
allows for fixed-parameter tractability. Furthermore, we present other algorithms
that, partially, allow us to successfully bypass the hardness presented above.

\section{(Parameterized) Tractability}
Our search for efficient algorithms for~\probMCSLong{} starts
with an introductory part presenting two lemmata exploiting structural
properties of problem solutions. Afterwards, we employ the
lemmata and provide several tractability results, including polynomial-time
solvability for one specific case.

\subsection{Identifying Structured Solutions}
A solution to the \probMCSLong{} is an ordered list of jobs to be executed on the
machine(s). Additionally, the jobs need to be associated with their starting times.
The starting times have to be
chosen in such a way that no job starts when the machine is still processing
another scheduled job and that each job requirement is met at the moment of
starting the job. We show that, in fact, given an order of jobs, one can always
compute the times of starting the jobs minimizing the makespan in polynomial
time. Formally, we present in~\Cref{lem:gapless} a polynomial-time Turing reduction
from~$1 | \nr=r | \cmax$ to~$ 1, \NI| \nr=r | -$.
The crux of this lemma is to observe that there always exists an optimal
solution to $1 | \nr=r | \cmax$ that is decomposable into two parts. First,
when the machine is idling, and second, when the machine is
continuously busy until all jobs are processed.

\begin{lemma}
 \label{lem:gapless}
 There is a polynomial-time Turing reduction from~$1|\nr=r|\cmax$
 to~\noGapProblem{}.
\end{lemma}
\begin{proof}
 Assuming that we have an oracle for~\noGapProblem{}, we describe an
 algorithm solving~$1|\nr=r|\cmax$ that runs in polynomial time.

 We first make a useful observation about
 feasible solutions to the original problem. Let us consider some feasible
 solution~$\sigma$ to~$1|\nr=r|\cmax$ and let~\orderedlistingof{g}{n} be the idle
 times before processing, respectively, the jobs~\orderedlistingof{J}{n}. Then,~$\sigma$
 can be transformed to another schedule~$\sigma'$ with the same makespan as~$\sigma$ and
 with idle times~\orderedlistingof{g'}{n} such that~$g'_2 = g'_3 \ldots g'_n=0$
 and~$g'_1=\sum_{t \in [n]}g_t$. Intuitively, $\sigma'$ is a scheduling in which
 the idle times of~$\sigma$ are all ``moved'' before the machine starts the
 first scheduled job. It is straightforward to see that in~$\sigma'$ no jobs are
 overlapping. Furthermore, each job according to~$\sigma'$ is processed at
 earliest at the same time as it is processed according to~$\sigma$. Thus,
 because there are no ``negative'' supplies and the order of processed jobs is
 the same in both~$\sigma$ and~$\sigma'$, each job's resource request is met in
 scheduling~$\sigma'$. 

 Using the above observation, the algorithm solving~$1|\nr=r|\cmax$ using an
 oracle for \noGapProblem{} problem works as follows. First, it guesses a
 starting gap's duration~$g \leq u_{\max}$ and then calls an oracle
 for~$\noGapProblem$ subtracting~$g$ from each supply time (and merging all
 non-positive supply times to a new one arriving at time zero) of the
 original $1|\nr=r|\cmax$~instance. For each value of~$g$, the algorithm adds~$g$
 to the oracle's output and returns the minimum over all these sums.

 Basically, the algorithm finds a scheduling with the smallest possible makespan
 assuming that the idle time happens only before the first scheduled job is
 processed. Note that this assumption can always be satisfied by the initial observation.
 Because of the monotonicity of the makespan with
 respect to the initial idle time~$g$, the algorithm can perform binary
 search while searching for~$g$ and thus its running time is~$O(\log(u_{\max}))$.
\end{proof}

Let us further explain the crucial observation backing~\Cref{lem:gapless}
since we will extend it in the subsequent \Cref{lem:ordered}.
Assume that, for some instance of the \probMCSLong{}, there is some optimal schedule
where some job~$J$ starts being processed at some time~$t$ (in particular, the
resource requirements of~$J$ are met at~$t$).
If, directly after the job the machine idles for some time, 
then we can
postpone processing~$J$ to the latest moment which still guarantees that~$J$
is ended before the next job is processed.
Naturally, at the new starting time of~$J$ we can
only have more resources than at the old starting time. Applying this
observation exhaustively produces a solution that is clearly separated into
idling time and busy time.

We will now further exploit the above observation beyond
only ``moving'' jobs without changing their mutual order. We first define
a~\emph{domination relation} over jobs; intuitively, a job dominates another
job if it is not shorter and at the same time it requires not more resources.
\begin{definition}\label{def:domination}
 A job~$J_j$ \emph{dominates} a job~$J_{j'}$ (written $J_j \le_D J_{j'}$)
 if~$p_j \geq p_{j'}$ and~$a_{i,j} \leq a_{i,j'}$ for all~$i \in \R$.
\end{definition}
When we deal with non-idling schedules, for a pair of jobs $J_j$
and~$J_{j'}$ where~$J_{j}$ dominates~$J_{j'}$, it is better (or at least not worse) to
schedule~$J_{j}$ before~$J_{j'}$.
Indeed, since among these two, $J_{j}$'s requirements are not greater and its processing time
is not smaller, surely after the machine stops processing~$J_j$ there will be at least as many
resources available as if the machine had processed~$J_{j'}$.
We formalize this observation in the following lemma.
\begin{lemma}\label{lem:ordered}
 For an instance of~$1, \NI| \nr | - $, let $<_D$ be an asymmetric subrelation of~$\le_D$. 
 There always is a feasible schedule where for every
 pair~$J_{j}$ and~$J_{j'}$ of jobs it holds that if $J_{j} <_D  J_{j'}$, then~$J_{j}$ is processed
 before~$J_{j'}$.
\end{lemma}
\newcommand{\Jin}{\ensuremath{\J_{\textrm{in}}}}
\newcommand{\Jout}{\ensuremath{\J_{\textrm{out}}}}
\begin{proof}
 Let~$\sigma$ be some feasible schedule.
 Consider a pair~$(J_{j},J_{j'})$ of jobs such that~$J_{j}$ is scheduled after~$J_{j'}$
 and~$J_{j}$ dominates~$J_{j'}$, that is, $p_j \geq p_{j'}$, $a_{ij} \leq a_{ij'}$
 for all~$i \in \R$.
 Denote by~$\sigma'$ a schedule emerging from continuously scheduling all
 jobs in the same order as in~$\sigma$ but with jobs~$J_{j}$ and~$J_{j'}$ swapped.
 Assume that~$\sigma'$ is not a feasible schedule.
 We show that each job in~$\sigma'$ meets its resource requirements, thus contradicting
 the assumption and proving the lemma.
 We distinguish between the set of jobs~$\Jout$ that are scheduled
 before~$J_{j'}$ in~$\sigma$ or that are scheduled after~$J_{j}$ in~$\sigma$ and
 jobs~$\Jin$ that are scheduled between~$j'$ and~$j$ in~$\sigma$
 (including~$j'$ and~$j$).
 Observe that since all jobs in~$\Jin$ are scheduled without the machine idling,
 it holds that all jobs in~$\Jout$ are scheduled exactly at the same times
 in both~$\sigma$ and~$\sigma'$.
 Additionally, since the total number of resources consumed by jobs in~$\Jin$
 in both~$\sigma$ and~$\sigma'$ is the same, the resource requirements for each
 job in~$\Jout$ is met in~$\sigma'$.
 It remains to show that the requirements of all jobs in~$\Jin$ are still met
 after swapping.
 To this end, observe that all jobs except for~$J_{j'}$ still meet the
 requirements in~$\sigma'$ as~$J_{j}$ dominates~$J_{j'}$ (i.e.,\ $J_{j}$~requires at most as
 many resources and has at least the same processing time as~$J_{j'}$).
 Thus, each job in~$\Jin$ except for~$J_{j'}$ has at least as many resources available
 in~$\sigma'$ as they have in~$\sigma$.
 Observe that~$J_{j'}$ is scheduled later in~$\sigma'$ than~$J_{j}$ was scheduled in~$\sigma$.
 Hence, there are also enough resources available in~$\sigma'$ to process~$J_{j'}$.
 Thus,~$\sigma'$ is feasible, a contradiction.
\end{proof}
\let\Jin\undefined
\let\Jout\undefined
 
Note that in the case of two jobs $J_{j}$ and $J_{j'}$ dominating each other
($J_{j} \le_D J_{j'}$ and~${J_{j'} \le_D J_{j}}$), \Cref{lem:ordered}
allows for either of them to be processed before the other one.

\subsection{Applying Structured Solutions}

We start with polynomial-time algorithms that applies both~\Cref{lem:gapless}
and~\Cref{lem:ordered} to solve special cases of the~\probMCSLong{} where each two
jobs can be compared according to the domination relation
(\Cref{def:domination}). Recall that if this is the case, then~\Cref{lem:ordered}
almost exactly specifies the order in which the jobs should be scheduled.
\begin{theorem}\label{thm:unit-P}
 $1, \NI | nr | -$ and~$1 | \nr | \cmax$ are solvable in, respectively, cubic and
	quadratic time if the domination relation is a weak order\footnote{A weak order of elements ranks elements such that each two objects are comparable but different objects can be tied.} on a set of jobs. In
 particular, for the time~$\umax$ of the last supply, $1|\nr=1,\ p_j = 1 | \cmax$
and $1|\nr=1,\ a_j = 1 | \cmax$ are solvable in $O(n \log n \log \umax)$ time and
 $1, \NI|\nr=1,\ p_j = 1| -$ and $1, \NI|\nr=1,\ a_j = 1| -$ are solvable in $O(n
 \log n)$~time.
\end{theorem}
\begin{proof}
 We first show how to solve~$1, \NI|\nr=1,\ p_j = 1| -$. At the beginning, we
 order the jobs increasingly with respect to their requirements of the
 resource arbitrarily ordering jobs with equal requirements.
 Then, we simply check whether scheduling the jobs in the computed order yields
 a feasible schedule, that is, whether the resource requirement of each job is met. If the
 check fails, then we return ``no,'' otherwise we report ``yes.'' The algorithm is
 correct due to~\Cref{lem:ordered} which, adapted to our case, says that there
 must exist an optimal schedule in which jobs with smaller resource requirements
 are always processed before jobs with bigger requirements. It is
 straightforward to see that the presented algorithm runs in~$O(n \log
 n)$~time.

 To extend the algorithm to~$1|\nr=1,\ p_j = 1 | \cmax$, we
 apply~\Cref{lem:gapless}. As described in detail in the proof
 of~\Cref{lem:gapless}, we first guess the idling-time~$g$ of the machine at the
 beginning. Then, we run the algorithm for~$1, \NI|\nr=1,\ p_j = 1| -$ pretending
 that we start at time~$g$ by shifting backwards by~$g$ the times of all
 resource supplies. Since we can guess~$g$ using binary search in a range
 from~$0$ to the time of the last supply, such an adaptation yields a
 multiplicative factor of~$O(\log u_{\max})$ for the running time of the algorithm
 for~$1, \NI|\nr=1,\ p_j = 1| -$. The correctness of the algorithm follows
 immediately from the proof of~\Cref{lem:gapless}.

 The proofs for~$1|\nr=1,\ p_j = 1 | -$ and $1|\nr=1,\ a_j = 1 | \cmax$ as well as
 the algorithms for these problems are analogous to those of~$1|\nr=1,\ p_j = 1 |
 -$ and $1|\nr=1,\ p_j = 1 | \cmax$.

 The aforementioned algorithms need only a small modification to work for the
 cases of~$1|\nr| -$ and ~$1|\nr| \cmax$ in which the domination relation is a
 weak order on a set of jobs. Namely, instead of sorting the jobs, one needs to
 find a weak order over the jobs. This is doable in quadratic time by comparing
 all pairs of jobs followed by checking whether the comparisons induce a
 weak order; thus we obtain the claimed running time. The obtained algorithm is
 correct by the same argument as the other above-mentioned cases.
\end{proof}

Importantly, it is simple (requiring at most~$O(n^2)$~comparisons) to
identify the cases for which the above algorithm can be applied successfully.

If the given jobs cannot be weakly ordered by domination,
then the problem becomes NP-hard as shown in~\Cref{thm:b-pNP-hard}.
This is to be expected since when jobs appear which are incomparable with respect to domination, then
one cannot efficiently decide which job, out of two, to
schedule first: the one which requires fewer resource units but has a shorter
processing time, or the one that requires more resource units but has a longer
processing time. Indeed, it could be the case that sometimes one may want to
schedule a shorter job with lower resource consumption to save resources for
later, or sometimes it is better to run a long job consuming, for example, all
resources knowing that soon there will be another supply with sufficient
resource units. Since NP-hardness probably excludes polynomial-time
solvability, we turn to a parameterized complexity analysis to get around the
intractability.

The time~$u_{\max}$ of the last supply seems a promising parameter.
We
show that it yields fixed-parameter tractability.
Intuitively, we demonstrate that the problem is tractable when the time until
all resources are available is short.

\begin{theorem} \label{thm:umax-fpt}
 $1, \NI|\nr=1|\cmax$ parameterized by the time~$u_{\max}$ of the
 last supply is fixed-parameter tractable and can be solved in~$O(2^{\umax}
 \cdot n + n \log n)$ time.
\end{theorem}
\begin{proof}
 We first sort all jobs by their processing time in~$O(n)$ time using bucket
 sort. We then sort all jobs with the same processing time by their
 resource requirement in overall~$O(n \log n)$ time.
 We then iterate over all subsets~$R$ of~$\{1,2,\ldots,\umax\}$. We
 will refer to the elements in~$R$ by~$r_1,r_2,\ldots,r_k$, where~$k = |R|$ and~$r_i < r_j$
 for all~$i < j$.
 For simplicity, we will use~$r_0 = 0$. For each~$r_i$
 in ascending order, we check whether there is a job with a processing
 time~$r_i - r_{i-1}$ that was not scheduled before and if so, then we schedule
 the respective job that is first in each bucket (the job with the lowest
 resource requirement). Next, we check whether there is a job left that
 can be scheduled at~$r_k$ and which has a processing time at least~$\umax -
 r_k$. Finally, we schedule all remaining jobs in an arbitrary order and check
 whether the total number of resources suffices to run all jobs.

 We will now prove that there is a valid gapless schedule if and only if all of
 these checks are met. Notice that if all checks are met, then our algorithm
 provides a valid gapless schedule.
 Now assume that there is a valid gapless schedule.
 We will show that our algorithm finds a (possibly different) valid gapless schedule.
 Let, without loss of generality, $J_{j_1},J_{j_2},\ldots,J_{j_n}$ be a valid
 gapless schedule and let~$j_k$ be the index of the last job that is
 scheduled latest at time~$\umax$. We now focus on the iteration where~$R =
 \{0, p_{j_1}, p_{j_1} + p_{j_2}, \ldots, \sum_{i=1}^k p_{j_i}\}$.
 If the algorithm schedules the jobs~$J_{j_1},J_{j_2},\ldots, J_{j_k}$,
 then it computes a valid gapless schedule and all checks are met.
 Otherwise, it schedules some jobs differently but, by construction, it always
 schedules a job with processing time~$p_{j_i}$ at  position~$i \leq k$.
 Due to \cref{lem:ordered} the schedule computed by the algorithm is also valid.
 Thus the algorithm computes a valid gapless schedule
 and all checks are met.

 It remains to analyze the running time. The sorting steps in the beginning
 take~$O(n \log n)$ time. There are~$2^{\umax}$ iterations for~$R$, each
 taking~$O(n)$ time. Indeed, we can check in constant time for each~$r_i$ which
 job to schedule and this check is done at most~$n$ times (as afterwards
 there is no job left to schedule). Searching for the job that is scheduled at
 time~$r_k$ also takes~$O(n)$ time as we can iterate over all remaining jobs and
 check in constant time whether it fulfills both requirements.
\end{proof}

Another possibility for fixed-parameter tractability via parameters
measuring the resource supply structure comes from combining the
parameters~$q$ and~$\bmax$.
Although both parameters alone yield intractability, combining them
gives fixed-parameter tractability in an almost trivial way:
By Assumption 1, every job requires at least one resource, so~${\bmax \cdot q}$ is an upper bound for the number of jobs.
Hence, with this parameter combination, we can try out all possible
schedules without idling (which by \Cref{lem:gapless} extends to
solving to $1, \NI|\nr=1|\cmax$).

Motivated by this, we replace the parameter $\bmax$~by the presumably much
smaller (and hence practically more useful) parameter~$\amax$.
We consider scenarios with only few resource supplies and jobs that
require only small units of resources as practically relevant.
Next, \Cref{thm:q-bmax-fpt} employs the technique of Mixed Integer
Linear Programming (MILP)~\cite{journal-BFNST20} to positively answer the question
of fixed-parameter tractability for the combined parameter~${q + \amax}$.

\begin{theorem} \label{thm:q-bmax-fpt}
 $1, \NI|\nr=1|\cmax$ is fixed-parameter tractable for the combined
 parameter $q+\amax$, where $q$~is the number of supplies and $\amax$~is the maximum
 resource requirement per job.
\end{theorem}
\begin{proof}
 Applying the famous theorem of \citet{Len83}, we describe an integer linear program
 that uses only $f(q,\amax)$ integer variables.
 \citet{Len83} showed that an (mixed) integer linear program is fixed-parameter tractable
 when parameterized by the number of integer variables (see also \citet{FT87} and \citet{Kan87}
 for later asymptotic running-time improvements).
 To significantly simplify the description of the integer program,
 we use an extension to integer linear programs that allows concave transformations
 on variables~\citep{journal-BFNST20}.
 
 Our approach is based on two main observations.
 First, by \cref{lem:ordered} we can assume that there is always
 an optimal schedule that is consistent with the domination order.
 Second, within a phase (between two resource supplies), every job can be arbitrarily reordered.
 Roughly speaking, a solution can be fully characterized by the number of jobs
 that have been started for each phase and each resource requirement.
 
 We use the following non-negative integer variables:
 \begin{enumerate}
  \item $x_{w,s}$ denoting the number of jobs requiring~$s$ resources
        started in phase~$w$,
  \item $x^{\Sigma}_{w,s}$ denoting the number of jobs requiring~$s$ resources
        started in all phases between $1$ and~$w$ (inclusive),
  \item $\alpha_w$ denoting the number of resources available in the beginning
        of phase~$w$,
  \item $d_w$ denoting the endpoint of phase~$w$, that is, the time when the last job started in phase~$w$ ends.
 \end{enumerate}
 
 Naturally, the objective is to minimize $d_q$. 
 First, we ensure that $x^{\Sigma}_{w,s}$ are correctly computed from $x_{w,s}$ by requiring
 $
  x^{\Sigma}_{w,s}= \sum_{w'=1}^{w} x_{w',s}.
 $
 Second, we ensure that all jobs are scheduled at some point.
 To this end, using~$\#_s$ to denote the number of jobs~$J_j$ with resource
 requirement~$a_j=s$, we add:
 $
   \forall s \in [\amax]: \sum_{w \in [q]} x_{w,s}=\#_s.
 $
 Third, we ensure that the~$\alpha_w$ variables are set correctly, by setting $\alpha_1=\tilde{b}_1$,
 and $\forall 2 \le w \le q:$
 $
  \alpha_w = \alpha_{w-1} + \tilde{b}_w - \sum_{s \in [\amax]} x_{w-1,s} \cdot s.
 $
 Fourth, we ensure that we always have enough resources:
 $
  \forall 2 \le w \le q: \alpha_w \ge \tilde{b}_w.
 $
 Next, we compute the endpoints~$d_w$ of each phase, assuming a schedule
 respecting the domination order.
 To this end, let $p^s_1$, $p^s_2$, $\dots$, $p^s_{\#_s}$ denote the processing times
 of jobs with resource requirement exactly~$s$ in non-increasing order.
 Further, let $\tau_s(y)$ denote the processing time spent to schedule the~$y$ longest jobs with resource requirement
 exactly~$s$, that is, we have $\tau_s(y)=\sum_{i=1}^{y}p^s_i$.
 Clearly, $\tau_s(x)$ is a concave function that can be precomputed for each $s \in [\amax]$.
 To compute the endpoints, we add:
 
 \begin{align} \label{const:ilp-convex}
  \forall w \in [q]: d_w=\sum_{s \in [\amax]} \tau_s(x^\Sigma_{w,s}).
 \end{align}
 
 Since we assume gapless schedules, we ensure that there is no gap:
 $
  \forall 1 \le w \le q-1: d_w \ge u_{w+1}-1.
 $
 This completes the construction of the mixed ILP using concave transformations.
 The number of integer variables used in the ILP is $2 q \cdot \amax$ (for
 $x_{w,s}^{(\Sigma)}$ variables) plus~$2q$ ($q$ for $\alpha_w$ and $d_w$
 variables, respectively).
 Moreover, the only concave transformations used in Constraint Set~\eqref{const:ilp-convex}
 are piecewise linear with only a polynomial
 number of pieces (in fact, the number of pieces is at most the number of jobs),
 as required to obtain fixed-parameter tractability of this extended class of
 ILPs~\citep[Theorem 2]{journal-BFNST20}.
\end{proof}

\section{A Glimpse on Multiple Resources}
So far we focused on scenarios with only one non-renewable resource.
In this section, we provide an outlook on scenarios with multiple resources
(still considering only one machine).
Naturally, all hardness results transfer.
For the tractability results, we identify several cases where tractability extends
in some form, while other cases become significantly harder.

Motivated by \Cref{thm:q-bmax-fpt}, we are interested in the computational complexity
of the \probMCSLong{} for cases where only~$\amax$ is small.
When $\nr=1$ and~$\amax=1$, then we have polynomial-time solvability via \Cref{thm:unit-P}.
The next theorem shows that this extends to the case of constant values of~$\nr$ and~$\amax$ if we assume unary encoding.
To obtain this results, we develop a dynamic-programming-based algorithm for
\noGapProblemOne{} and apply \Cref{lem:gapless}.
\begin{theorem}\label{thm-amax-XP}
	$1|\nr=\const,\unary|\cmax$ can be solved in~$O(q \cdot \amax^{O(1)} \cdot n^{O(\amax)} \cdot \log u_{\max})$ time.
\end{theorem}
\begin{proof}
	We will describe a dynamic-programming procedure that computes whether there exists a gapless schedule (\Cref{lem:gapless}).
	Let~$r$ be the (constant) number of different resources.
	We distinguish jobs by their processing time as well as their resource requirements.
	To this end, we define the \emph{type} of a job~$J_j$ as a vector~$(a_{1,j},a_{2,j},\cdots,a_{r,j})^T$ containing all its resource requirements.
	Let~$\mathcal{T} = \{t_1,t_2,\ldots,t_{|\mathcal{T}|}\}$ be the set of all types such that for all~$t \in \mathcal{T}$ there is at least one job~$J_j$ of type~$t$.
	Let~$s := |\mathcal{T}|$ and note that~$s \leq (\amax+1)^r$.
	We first sort all jobs by their type using bucket sort in~$O(n)$ time and then sort each of the buckets with respect to the processing times in~$O(n \log n)$ time using merge sort.
	For the sake of simplicity, we will use~$P_t[k]$ to describe the set of the $k$~longest jobs of type~$t$.
	We next define the \emph{$\ell$\textsuperscript{th} phase} as the time interval from the start of any possible schedule up to~$u_{\ell+1}$.
	Next, we present a dynamic program $$T \colon [q] \times [n_{t_1}] \cup \{0\} \times \ldots \times [n_{t_s}] \cup \{0\} \rightarrow \{true,false\},$$ where~$n_{t_i}$ is the number of jobs of type~$t_i$.
	We want to store~$true$ in~$T[i,x_1,\ldots,x_s]$ if and only if it is possible to schedule at least the jobs in~$P_{t_k}[x_k]$ such that all of these jobs start within the~$i$\textsuperscript{th} phase and there is no gap.
	If~$T[q,n_{t_1},\ldots,n_{t_s}]$ is true, then this corresponds to a gapless schedule of all jobs and hence this is a solution.
	We will fill up the table~$T$ by increasing values of the first argument.
	That is, we will first compute all entries~$T[1,x_1,x_2,\ldots,x_s]$ for all possible combinations of values for~$x_i\in [n_{t_i}] \cup \{0\}$.
	For the first phase, observe that~$T[1,x_1,x_2,\ldots,x_s]$ is set to true if and only if the two following conditions are met.
	First, there are enough resources available at the start to schedule all jobs in all~$P_{t_i}[x_i]$.
	Second, the sum of all processing times of all ``selected jobs'' without the longest one end at least one time steps before~$u_2$ (such that the job with the longest processing time can then be started at last in time step~$u_2 -1$ which is the last time step in the first phase).
	For increasing values of the first argument, we do the following to compute~$T[i,x_1,x_2,\ldots x_s]$.
	We compute for all tuples of numbers~$(y_1,y_2,\ldots y_s)$ with~$y_k \leq x_k$ for all~$k \in [s]$ whether~$T[i-1,y_1,y_2,\ldots,y_s]$ is true, whether the corresponding schedule can be extended to a schedule for~$T[i,x_1,x_2,\ldots,x_{s}]$, and whether all selected jobs except for the longest one can be finished at least two time steps before~$u_{i+1}-1$.
	Since the first check and third check are very simple, we focus on the second check.
	To this end, it is actually enough to check whether~$\sum_{k\in [s]} \sum_{j \in P_{t_k}[y_k]} p_j \geq u_i-1$ since if this was not the case but there still was a gapless schedule, then we could add some other job to the~$(i-1)$\textsuperscript{st} phase and since we iterate over all possible combinations of values of~$y_i$, we would find this schedule in another iteration.
	
	It remains to analyze the running time of this algorithm.
	First, the number of table entries in~$T$ is upper-bounded by~$q \cdot \prod_{k\in [s]} (n_{t_k}+1) \leq (q \cdot (n+1)^s)$.
	For each table entry~$T[i,x_1,x_2,\ldots,x_s]$, there are at most~$\prod_{k\in [s]} (x_k+1) \leq  \prod_{k\in [s]} (n_{t_k}+1) \in O((n+1)^{s})$ possible tuples~$(y_1,y_2,\ldots,y_s)$ with~${y_k \leq x_k}$ for all~$k \in [s]$ and the three checks can be performed in~$O(s)$ time.
	Thus, the overall running time is~$O(q \cdot s \cdot (n+1)^{2s}) \subseteq O(q \cdot (\amax+1)^r \cdot (n+1)^{2r(\amax+1)})$ for computing a gapless schedule and by \Cref{lem:gapless} the time for solving~$1|nr=\const,\unary|\cmax$ is in~$O(q \cdot (\amax+1)^r \cdot (n+1)^{2r(\amax+1)} \cdot \log \umax)$.
\end{proof}
The question whether $1|\nr=\const,\unary|\cmax$ is in \FPT or \wonehard{} with respect to~$\amax$ remains open even for only a single resource.

We continue
with showing that already with two resources and
unit processing times of the jobs, the \probMCSLong{} becomes computationally
intractable, even when parameterized by the number of supply dates.
Note that \NP-hardness for $1|\nr=2,p_j=1|\cmax$ can also be transferred from
\citet[Theorem 4]{Grigoriev05} (the statement is for a different optimization
goal but the proof works).

\begin{proposition}
\label{prop:nr2qW1}
 $1|\nr=2,p_j=1|\cmax$ is \wone-hard when parameterized by the number of supply
 dates even if all numbers are encoded in unary.
\end{proposition}

\begin{proof}
 Given an instance~$I$ of \textsc{Unary Bin Packing} with
 $k$~bins, each of size~$B$, and $n$~objects~$o_1,o_2,\ldots o_n$ of
 sizes~$s_1,s_2,\ldots s_n$ such that~$\sum_{i=1}^{n} s_i = k B$, we construct an
 instance~$I'$ of~$1|\nr=2,p_j=1|\cmax$ as follows.
 
 For
 each~$i \in [k]$, we add a supply~$q_i = (u_i, \tilde{b}_{1,i},
 \tilde{b}_{2,i}) := ((i-1)B, B, B(B-1))$; thus, we create $k$~supply dates. For
 each object~$o_i$ for~$i \in [n]$, we create an~\emph{object job}~$j_i =
 (p_i, a_{1,i}, a_{2,i}) := (1, s_i, B-s_i)$. Additionally, we create $kB -
 n$~\emph{dummy jobs}; each of them having processing time~$1$, no
 requirement of the resources of the first type, and requiring $B$~resources of
 the second type. We refer to the constructed instance as~$I'$.

 We show that~$I$ is a yes-instance if and only if there is a schedule for the
 constructed instance~$I'$ with makespan exactly~$kB$. Clearly, if $I$~is a
 yes-instance, then there is a $k$-partition of the objects to subsets~$S_1,
 S_2, \ldots, S_k$ such that for each~$i \in k$, $\sum_{o_j \in S_i} s_j = B$.
 For some set~$S_i$, we schedule all jobs in~$S_i$, one after
 another, starting at time~$(i-1)B$. Then, in the second step, we schedule the
 remaining dummy jobs such that we obtain a gapless
 schedule~$\sigma$. Naturally, since each~$S_i$ contains at most~$B$~objects,
 the object jobs are non-overlapping in~$\sigma$. Since in the second step we
 have exactly~$kB - n$~jobs available (recall that~$n$~is the number of object
 jobs) $\sigma$~is a gapless schedule with makespan exactly~$kB$. To check the
 feasibility of~$\sigma$, let us consider the first~$B$~time steps. Note that
 from time~$0$ to time~$|S_1|$, the machine processes all object jobs
 representing objects from~$S_1$; from time~$|S_1|+1$ to~$B$ it processes
 exactly $B-|S_1|$~dummy jobs. Thus, the number of used resources of type~$1$
 is~$\sum_{o_j \in S_i} s_j = B$, and the number of used resources of type~$2$
 is~$\sum_{o_j \in S_i} (B-s_j) + (B-|S_1|)B = |S_1|B - B + B^2 - |S_1|B =
 B(B-1)$. As a result, at time~$B-1$ there are no resources available, so we can
 apply the argument for the first~$B$~time steps to all following~$k-1$~periods
 of~$B$~time steps; we eventually obtain that $\sigma$~is a feasible schedule.
 
 \newcommand{\phaseJobs}{\ensuremath{J_{\textrm{p}}}}
 \newcommand{\objectPhaseJobs}{\ensuremath{J_{\textrm{o}}}}
 For the reverse direction, let~$\sigma$~be a schedule with makespan~$kB$ for
 instance~$I'$. We again consider the first~$B$~time steps. Let~\phaseJobs{} be the set of
 exactly~$B$~jobs processed in this time. Let~$A_1$ and~$A_2$ be the usage of
 the resource of, respectively, type~$1$ and~type~$2$ by the jobs
in~\phaseJobs{}. We denote by~$\objectPhaseJobs{} \subseteq \phaseJobs$ the
object jobs within~$\phaseJobs$. Then, $A_1 := \sum_{j \in
 \objectPhaseJobs} a_{1,j} + \sum_{j \in \phaseJobs \setminus \objectPhaseJobs}
 a_{1,j}$. In fact, since each dummy job has no requirement of the resources of
 the first type, we have~$A_1 = \sum_{j \in \objectPhaseJobs} a_{1,j}$.
 Moreover, there are only $B$~resources of type~$1$ available in the
 first~$B$~time steps, so it is clear that~$A_1 \leq B$.
 Using the fact that, for each job~$j_i$, it holds that~$a_{2,i} = B - a_{1,i}$,
 we obtain the following:
 \begin{align*}
 A_2 &:= \sum_{j \in \objectPhaseJobs} a_{2,j} + \sum_{j \in \phaseJobs
 \setminus \objectPhaseJobs} a_{2, j} \\&= |\objectPhaseJobs|B - A_1 + |\phaseJobs
 \setminus \objectPhaseJobs| B = B^2 - A_1.
 \end{align*}

 Using the above relation between~$A_1$ and~$A_2$, we show that~$A_1=B$. For the
 sake of contradiction, assume $A_1<B$. Immediately, we obtain that~$A_2 > B^2 -
 B = B(B - 1)$, which is impossible since we only have~$B(B - 1)$~resources of
 the second type in the first~$B$~time steps of schedule~$\sigma$. Thus,
 since~$A_1 = B$, we use exactly~$B$~resources of the first type and
 exactly~$(B-1)B$~resources of the second type in the first~$B$~time~steps
 of~$\sigma$. We can repeat the whole argument to all following~$k-1$~periods of
 $B$~time steps. Eventually, we obtain a solution to~$I$ by taking the objects
 corresponding to the object jobs scheduled in the subsequent periods of
 $B$-time steps.

 The reduction is clearly applicable in polynomial time and the number of supply
 dates is a function depending solely on the number of bins in the input instance
 of~\unarybinpacking{}.
\end{proof}

\Cref{prop:nr2qW1} limits the hope for obtaining positive results
for the general case with multiple resources.
Still, when adding the number of different resources to the combined 
parameter, we can extend
our fixed-parameter tractability result from \Cref{thm:q-bmax-fpt}.
Since we expect the number of different resources to be rather small in real-world applications,
we consider this result to be of practical interest.

\begin{proposition}\label{prop:q-bmax-fpt-const-resources}
 $1, \NI|\nr=r|\cmax$ is fixed-parameter tractable for the parameter~${q + \amax + r}$, where~$q$~is the number of supplies and~$\amax$~is the maximum
 resource requirement of a job.
\end{proposition}
\begin{proof}
 The main observation needed to extend the ILP from \Cref{thm:q-bmax-fpt} is
 that, by \Cref{lem:ordered}, given two jobs with the same resource requirement,
 there is always a schedule that first schedules the longer (dominating) jobs.
 In essence, for each phase and each possible resource requirement, a solution
 is still fully described by the respective number of jobs with that requirement
 scheduled in the phase.
 
 For multiple resources, we describe the resource requirement of job~$j$ by a
 \emph{resource vector}~${\vec{s}=(a_{1,j},a_{2,j},\ldots,a_{r,j})}$.
 We use the following non-negative integer variables:
 \begin{enumerate}
  \item $x_{w,\vec{s}}$ denoting the number of jobs with resource vector~$\vec{s}$
        being started in phase~$w$,
  \item $x^{\Sigma}_{w,\vec{s}}$ denoting the number of jobs with resource vector~$\vec{s}$
        being started between phase~$1$ and~$w$,
  \item $\alpha_{y,w}$ denoting the number of resources of type~$y$ available in the beginning
        of phase~$w$,
  \item $d_w$ denoting the endpoint of phase~$w$, that is, the time when
        the job started latest in phase~$w$ ends.
 \end{enumerate}
 
	All constraints and proof arguments translate in a straightforward way from the proof of \cref{thm:q-bmax-fpt}.
\end{proof}

Next, by a reduction from \textsc{Independent Set} we show that the \probMCSLong{} is intractable for an unbounded number of resources
even when combining all considered parameters.

\begin{theorem}\label{thm-superhard}
 $1|\nr,p_j=1|\cmax$ is \np-hard and \wonehard{} parameterized by~$\umax$
 even if $\pmax=\amax=\bmax=1$ and~$q=2$.
\end{theorem}
\begin{proof}
 We provide a parameterized reduction from the NP-hard \textsc{Independent Set}
 problem which, given an undirected graph~$G$ and a positive integer~$k$, asks
 whether there is an \emph{independent set} of size~$k$, that is, a set of $k$~vertices
 in~$G$ which are pairwise non-adjacent.
 \textsc{Independent Set} is \wonehard{} for the size~$k$ of the independent 
	set~\cite{DF13}.
 
 Given an \textsc{Independent Set} instance~$(G,k)$ we create an instance of $1|\nr,p_j=1|\cmax$
 as follows. Let $V(G)=\{v_1,\ldots,v_n\}$ and $E(G)=\{e_1,\ldots,e_m\}$.
 For each edge in~$G$, we create one resource, that is, $\nr=m$.
 At time~$u_1=0$, we provide one initial unit of every resource.
 At time~$u_2=k$, we provide another unit of every resource.
 For each vertex $v_j \in V(G)$, there is one job~$J_j$ of length one
 with resource requirement being consistent with the incident edges,
 that is, $a_{i,j}=1$ if $v_j \in e_i$ and $a_{i,j}=0$ if $v_j \notin e_i$.
 
 We claim that there is a schedule with makespan~$\cmax=n$ if and only if there is
 an independent set of size~$k$.
 
 For the ``if'' direction, let $V'=\{v_{\ell_1},\ldots,v_{\ell_k}\}$~be an
 independent set in~$G$.
 Observe that every schedule that schedules the jobs~$J_{\ell_1}$, $\ldots$, $J_{\ell_k}$
 (in any order) in the first phase and then all other jobs (in any order)
 in the second phase is feasible and has makespan~$\cmax=n$.
 Feasibility comes from the fact that we have an independent set so that no two
 jobs scheduled in the first phase require the same resource.
 (After the second supply, there are enough resources to schedule all jobs.)
 The makespan can be verified by observing that we have unit processing time
 and that the machine does not idle.
 
 For the ``only if'' direction, assume that there is a feasible schedule
 with makespan~$\cmax=n$ and let $\mathcal{J}'=\{J_{j_1},\ldots,J_{j_k}\}$ be the
 jobs which are scheduled at time points~$0$ to~$k-1$.
 (Note that~$\mathcal{J}'$ is well-defined since each job has length one,
  there are $n$~jobs in total, and the makespan is~$n$.)
 We claim that the set~$V'=\{v_{j_1},\ldots,v_{j_k}\}$ is an independent set.
 Assume towards a contradiction that two vertices are adjacent.
 Then, both jobs are scheduled in the first phase and require one unit
 of the the same edge resource; a contradiction.
\end{proof}

Finally, to complement \cref{thm-superhard}, we show that~$1|\nr|\cmax$ parameterized by~$\umax$ is in XP.
Note that for this algorithm we do not need to assume unit processing times.

\begin{proposition}\label{prop-umax-XP}
	$1|\nr|\cmax$ can be solved in~$O(n^{\umax+1} \cdot \log \umax)$ time.
\end{proposition}
\begin{proof}
	We solve~$1|\nr|\cmax$ by basically brute-forcing all schedules up to time~$\umax$.
	By \Cref{lem:gapless}, we may assume that we are looking for a gapless schedule.
	The algorithm now iteratively guesses the next job in the schedule (starting with a first job at time~$0$).
	Since we assume that processing any job requires at least one time unit, the algorithm guesses at most~$\umax$ jobs until time~$\umax$.
	Afterwards, we can schedule the jobs in any order as no new resources become available afterwards.
	Notice that guessing up to~$\umax$ jobs take~$O(n^\umax)$ time, verifying whether a schedule is feasible takes~$O(n)$ time, and \Cref{lem:gapless} adds an additional factor of~$\log \umax$ for assuming a gapless schedule.
	This results in an overall running time of~$O(n^{\umax+1} \cdot \log \umax)$.
\end{proof}

\section{Conclusion}
We provided a seemingly first thorough multivariate complexity analysis 
of the \probMCSLong{} on a single machine.
Our main concern was the
case of one resource type ($\nr=1$).
\Cref{tab:results} surveys our results.

Specific questions refer to the parameterized complexity with 
respect to the single parameters $\amax$ and~$\pmax$, their combination,
and the closely related parameter number of job types.
Notably, this might be challenging to answer 
because these questions are closely 
related to long-standing open questions for \textsc{Bin Packing}
and $P||\cmax$ \cite{KK18,KKM20,MB18}.
Indeed, parameter combinations may be unavoidable in order to identify practically
relevant tractable cases.
For example, it is not hard to derive from our statements (particularly \Cref{asp:1} and \cref{lem:gapless})
fixed-parameter tractability for $\bmax+q$ while for the single parameters~$\bmax$ and~$q$ it is both times computationally hard.

Another challenge is to study the case of multiple machines, which is
obviously computationally at least as hard as the case of a single machine but possibly very relevant in practice.
It seems, however, far from obvious to generalize our algorithms to the multiple-machines case.

We have also seen that cases where the jobs
can be ordered with respect to the domination ordering (\Cref{def:domination})
are polynomial-time solvable. It may be promising to consider structural
parameters measuring the distance from this tractable special case 
in the spirit of 
distance from triviality parameterization~\cite{GHN04,Nie06}.

Our results for multiple resources certainly mean only first steps.
They clearly invite 
to further investigations, particularly concerning a multivariate 
complexity analysis.

\section*{Acknowledgments}
The main work was done while Robert Bredereck was with TU~Berlin.
Andrzej Kaczmarczyk was supported by the DFG project AFFA (BR~5207/1 and NI~369/15).
Péter Györgyi was supported by the National Research, Development and Innovation
Office – NKFIH (ED\_18-2-2018-0006), and by the J\'{a}nos Bolyai Research Scholarship of the Hungarian Academy of Sciences.
The project started while Robert Bredereck, Péter Györgyi, and Rolf Niedermeier were attending the Lorentz center workshop
``Scheduling Meets Fixed-Parameter Tractability'', February 4--8, 2019, Leiden,
the Netherlands, organized by Nicole Megow, Matthias Mnich, and Gerhard~J.\ Woeginger. Without this meeting, this work would not have been started.

\bibliographystyle{plainnat}
\bibliography{mybibfile}

\end{document}